\documentclass [11pt] {article}

 \usepackage{fullpage}

\usepackage{graphics}
\usepackage[dvips]{epsfig}

\usepackage{amsmath}
\usepackage{amssymb}
\usepackage{amsfonts}
\usepackage{graphicx}
\usepackage{subfig}

\usepackage{cite}

\usepackage{algorithmic}
\usepackage[linesnumbered,ruled]{algorithm2e}

\usepackage{array}
\usepackage{color}

\begin{document}

\newtheorem{theorem}{Theorem}[section]
\newtheorem{lemma}{Lemma}[section]
\newtheorem{corollary}{Corollary}[section]
\newtheorem{claim}{Claim}[section]
\newtheorem{proposition}{Proposition}[section]
\newtheorem{definition}{Definition}[section]
\newtheorem{fact}{Fact}[section]
\newtheorem{example}{Example}[section]
\newtheorem{remark}{Remark}[section]
\newtheorem{case}{Case}

\newcommand{\cV}{{\cal V}}
\newcommand{\cA}{{\cal A}}

\newcommand{\qed}{\hfill $\square$ \smallbreak}
\newenvironment{proof}{\noindent{\bf Proof:}}{\qed}

\def\thefootnote{\fnsymbol{footnote}}

\title{{\bf Rendezvous of Mobile Deterministic Automata\\ in Graphs}}\date{}

\author{Bibhuti Das\footnotemark[1]
\and
Andrzej Pelc\footnotemark[2]
}

\footnotetext[1]{D\'epartement d'informatique, Universit\'e du Qu\'ebec en Outaouais, Gatineau,
Qu\'ebec J8X 3X7, Canada. {\tt dasbibhuti905@gmail.com}}

\footnotetext[2]{
 D\'epartement d'informatique, Universit\'e du Qu\'ebec en Outaouais, Gatineau,
Qu\'ebec J8X 3X7, Canada. {\tt pelc@uqo.ca}. Partially supported by NSERC discovery grant  2024-03767
and by the Research Chair in Distributed Computing at the
Universit\'e du Qu\'ebec en Outaouais.}

\maketitle

\thispagestyle{empty}

\begin{abstract}
Two mobile agents, modeled as identical deterministic finite automata (DFA) navigating in synchronous rounds in a graph with unlabeled nodes, have to meet at some node. The well-researched task of meeting in a graph is known as {\em rendezvous}.
Agents start at adversarially chosen distinct nodes in possibly different rounds. An instance of the rendezvous problem is the underlying graph, together with the initial nodes $u$ and $v$ of the agents. Such an instance is {\em feasible}, if there exists a DFA (possibly working only for this instance), such that its identical copies starting at nodes $u$ and $v$, with an arbitrary delay, accomplish rendezvous. A DFA is {\em  RV-universal} for a class of instances, if it guarantees rendezvous of its copies starting at the designated nodes with arbitrary delay, for all feasible instances of this class. Our goal is to investigate the existence of RV-universal DFA.

We start by observing that if agents cannot mark nodes in any way then there does not exist a RV-universal DFA even for the class of instances where the underlying graph is a line. Hence we allow the use of identical {\em pebbles} to mark the nodes by the agents. We consider {\em stationary pebbles} that can be dropped by agents at nodes but cannot be picked up, and {\em movable pebbles} that can be dropped by agents and later picked up. We observe that, even in the more powerful scenario of movable pebbles, if agents are equipped with any finite number of pebbles, there is no RV-universal DFA for the class of all instances. Hence we restrict attention to instances where the underlying graph is a tree. Our main contribution are two contrasting results showing that movability of pebbles is a crucial feature. We first prove that for any finite number of stationary pebbles there is no RV-universal DFA for trees, and then we design a RV-universal DFA for trees, where each agent is equipped with a single movable pebble.
\vspace*{0.2cm}

\noindent
{\bf keywords:} anonymous graph, anonymous mobile agent, deterministic finite automaton, pebble, rendezvous
\vspace*{1.5cm}
\end{abstract}

\pagebreak

\section{Introduction}

Two mobile agents, modeled as identical deterministic finite automata (DFA) navigating in synchronous rounds in a graph, have to meet at some node. The task of meeting of agents in a graph is known as {\em rendezvous}. It has been extensively researched, including in papers \cite{CKP1} (SPAA 2012), \cite{PY} (SPAA 2019) and \cite{DP} (SPAA 2026),
due to its fundamental symmetry-breaking character. Even if agents have unbounded memory and computational power, if they are identical, meeting is often impossible, e.g.,  when they start simultaneously in an anonymous oriented ring. Hence, symmetry between the agents is sometimes broken by assuming that they know their initial location in the graph \cite{CCGKM}, or by assigning them different labels \cite{TSZ07}.

In applications, mobile agents may be software agents navigating in a network in order to consult 
a distributed database whose parts are located in nodes of the network. They can also be mobile robots navigating in a network of corridors of a contaminated mine, too dangerous for human access, with the aim of collecting samples of the ground or of the air. In both cases, rendezvous may be needed to exchange data collected so far and plan future actions. Since mobile agents are often small and simple mass-produced devices with limited memory and computational power, the  choice of a DFA to model them, and the assumption that they are identical, seem natural. This is why we make this choice and assumption in the present paper.
 
Nodes of the graph are unlabeled, but ports at each node of degree $d$ are labeled by integers $0,1,\dots, d-1$.  Agents start at adversarially chosen distinct nodes in possibly different rounds. 
In each round, an agent may stay put at the current node or move to an adjacent node by choosing some port.

An instance of the rendezvous problem is the underlying graph, together with two distinct nodes that are initial positions of the agents. Such an instance $(G,u,v)$ is {\em feasible}, if there exists a DFA (possibly working only for this instance), such that its identical copies starting at nodes $u$ and $v$, with an arbitrary delay, accomplish rendezvous. Saying that the agents are identical copies of a DFA we understand that the two agents are modeled by the same DFA.
A simple example of a non-feasible instance is a line of odd length with symmetric port labeling, and agents starting at its extremities. No automaton $A$ (in fact even no Turing machine) may permit rendezvous for this instance if both agents modeled by $A$ start simultaneously. 
In contrast, a line of odd length and non-symmetric port labeling gives a feasible
instance, regardless of the initial positions of the agents. Indeed,  a sufficiently  large automaton could
memorize the line and guarantee rendezvous.

A DFA is {\em  RV-universal} for a class of instances, if it guarantees rendezvous of its copies starting at the designated nodes with arbitrary delay, for all feasible instances of this class. Our goal is to investigate the existence of RV-universal DFA. Such a universal automaton for a large and important class of instances permits us to accomplish the desired task of rendezvous, without knowing in which (feasible) instance of the considered class
the agents are navigating.

\subsection
{The model and the problem}
We consider finite simple undirected connected graphs $G=(V,E)$.
Nodes of the graph are unlabeled, but ports at each node of degree $d$ are labeled by integers $0,1,\dots, d-1$. There is no coherence between port numbers at two extremities of an edge.
Since we are interested in deterministic finite automata navigating in {\em a priori} unknown graphs, the automaton must be able to distinguish all port numbers, in order to transit to a state instructing it to take one of the ports. Hence we assume that all node degrees of considered graphs are bounded by some common constant $\ell$ (otherwise no single DFA could navigate in all such graphs), and refer
to finite simple undirected connected graphs of bounded degree simply as {\em graphs}.
Since a RV-universal DFA is considered for a class of graphs with node degrees bounded by  $\ell$, in the construction of such an automaton we assume the knowledge of $\ell$.

Mobile agents are represented as entities  
whose actions are modeled by the same deterministic finite automaton. Hence agents are undistinguishable.
They navigate in synchronous rounds. They start from two distinct adversarially chosen nodes of the graph, called their {\em bases}, in the initial state $S_0$, possibly in different rounds, also chosen by an adversary. In the beginning, each agent is equipped with $\rho$ pebbles, and all pebbles are identical. (Later on, we explain why agents must use pebbles). We consider two scenarios: pebbles can be either {\em stationary} or {\em movable}.  

In each round, an agent is at some node of the graph, carries some number $a$ of pebbles and is in some state $S$. It gets the following input: the current number $a\geq 0$ of pebbles it carries, the degree of the current node, the port number by which it entered this node, or $-1$, if it stayed put in the current round, and a bit indicating if there is a pebble at the current node. The state $S$ causes the agent to produce an output that has two parts: a {\em pebble action} that can be $drop$, $pick$ or $leave$ and a {\em move action} that can be either $-1$ or a non-negative integer. Seeing a given input in state $S$, the agent transits to some state $S'$. The output instructs the agent what to do in the next round. 

The pebble action $leave$ means that the agent does not  drop or pick up any pebble.
The pebble action $drop$ consists of dropping a pebble at the current node and decreasing the number of currently carried pebbles by 1. It can be executed only if $a>0$ and there is no pebble at the current node; otherwise this action has the same result as $leave$.
The pebble action $pick$ consists of picking the pebble at the current node and increasing the number of currently carried pebbles by 1. It can be executed only
if there is a pebble at the current node, and only in the scenario of movable pebbles; 
otherwise this action has the same result as $leave$.

The move action $-1$ means that the agent stays put in the next round, and the move action $m\geq 0$
means that the agent takes port $m$ in the next round. If there is no port $m$ at the current node, the agent stays put.
The trajectory of the agent is the infinite sequence of above described steps.

We now give a formalization of the above intuitive description.
Let $\{0,1\}$ be the set of two possibilities for the current node $z$ of the agent: $0$ indicates that there is no pebble in $z$ ($z$ is empty), and $1$ indicates that there is a pebble in $z$ ($z$ is full).
Let $PebbleActions=\{pick, drop, leave\}$.
Let $P=\{0,1,\dots,\rho\}$ and $L=\{-1,0,1,\dots,\ell\}$.

The actions of the mobile agent are formalized as a deterministic finite Mealy automaton
${\cal A}=(X,Y,Q,\delta,\lambda,S_0)$. 
$X=L\times L \times \{0,1\} \times P$ is the input alphabet and
$Y=PebbleActions\times L \times P$ is the output alphabet. $Q$ is a finite set of states
with a special state $S_0$ called initial.
$\delta:Q\times X \to Q$ is the state transition function, and $\lambda:Q  \to
Y$ is the output function.

The meaning of the input and output symbols is the following.  At each step of its functioning, the agent is at some node $z$ of the graph, in some state $S$,
and carries some number $m \in P$ of pebbles. 
It sees the degree $d\in L$ of the current node and the integer $i\in L$ which is $-1$ or the entry port number. It also learns the bit $h\in \{0,1\}$ indicating if there is a pebble at the current node.
The input $I=(d,i,h,m)\in X$ gives the automaton information about these facts. 
 Given the state $S$,
the agent outputs the symbol $\lambda(S)\in PebbleActions \times L  \times P$ with the following meaning. The first term indicates the pebble action performed by the agent in the current step. The second term indicates the port number the agent should take or the instruction to stay put if this term is $-1$, or if the given port number is not available. The third term indicates how many pebbles the agent carries in the next step. 

Since pebble actions are linked to the number of carried pebbles, 
we have the following restrictions on the possible values of the output function $\lambda$. 
Let $I=(d,i,h,m)\in X$, and $\lambda(S)=(\alpha, j, \kappa)$, where $\alpha \in PebbleActions$, $d\in L$ and $\kappa\in P$. 
\begin{itemize}
\item
if $\alpha=leave$ then $\kappa=m$;
\item
if $\alpha=drop$ and $h=0$ then $\kappa=m-1$; 
\item
if $\alpha=pick$, $h=1$ and the scenario is that of movable pebbles then $\kappa=m+1$.
\end{itemize}

 Given the input symbol $I$ and being in a current state $S$, the agent makes the changes indicated by the output function (it possibly drops or picks a pebble as indicated, possibly changes the number of carried pebbles as indicated, and makes the indicated move), and transits to state $\delta(S,I)$.
 The agent starts with $\rho$ pebbles in the initial state $S_0$, and its starting node is empty (hence its initial input symbol is $(-1,0,\rho)$).

\subsection{Our contribution}

We start by observing that if agents cannot mark nodes in any way then there does not exist a RV-universal DFA even for the class of instances where the underlying graph is a line. Hence we allow the use of identical {\em pebbles} to mark the nodes by the agents. We consider {\em stationary pebbles} that can be dropped by agents at nodes but cannot be picked up, and {\em movable pebbles} that can be dropped by agents and later picked up. We observe that, even in the more powerful scenario of movable pebbles, if agents are equipped with any finite number of pebbles, there is no RV-universal DFA for the class of all instances. Hence we restrict attention to instances where the underlying graph is a tree. Our main contribution are two contrasting results showing that movability of pebbles is a crucial feature. We first prove that for any finite number of stationary pebbles there is no RV-universal DFA for trees, and then we design a RV-universal DFA for trees, where each agent is equipped with a single movable pebble. Our negative result holds even for the class of instances whose underlying graph is a line.

\subsection
{Related work}

Rendezvous of mobile agents in graphs is an extensively studied topic in the distributed computing literature. Many scenarios have been adopted, concerning the capabilities and the behavior of the agents.
In the case when more than two agents are required to meet, rendezvous is often called gathering.

In most of the papers on rendezvous, it is assumed that nodes of the graph do not have distinct identities, and agents cannot mark nodes. However, scenarios departing from these assumptions were also studied.
Rendezvous was considered in graphs whose nodes are labeled \cite{CCGKM,MP}, or when marking nodes by agents using pebbles is allowed \cite{KKSS}.
For a survey of  randomized rendezvous we refer to the classic book
\cite{alpern02b}. 
Deterministic rendezvous in graphs was surveyed in \cite{Pe2}.
 
Most of the literature on rendezvous considered the synchronous scenario, where
agents move in rounds \cite{CCGKM,CKP,TSZ07}. We follow this assumption in the present paper.
 
 In many graphs, e.g., in the oriented ring, symmetry between agents must be broken because otherwise rendezvous is impossible. Most often symmetry is broken by assigning different labels to agents \cite{DFKP,TSZ07}. Sometimes an even stronger assumption is used:
 in \cite{CCGKM}, the authors considered rendezvous under the assumption that the agents know their location in the graph (then the initial location can serve as a label).
 If agents are anonymous and cannot mark nodes then rendezvous is possible only for restricted classes of instances \cite{CKP}.

In many papers, agents are modeled as Turing machines and their memory is unbounded. Other studies concern 
the minimum amount of memory that agents must have in order to accomplish rendezvous \cite{CKP,FP}. In this case, agents are modeled as state machines and the number of states is a function of the size of graphs in which they operate. Gathering many finite automata with distinct labels on the infinite line was considered in \cite{GP}.
To the best of our knowledge, rendezvous of anonymous deterministic finite  automata on anonymous graphs, that we consider in the present paper, was never studied before.


In many papers, asynchronous gathering and rendezvous was studied in the plane \cite{CFPS,fpsw} and 
in graphs
	\cite{BBDDP,BCGIL,DPV}. In the plane, agents are modeled as moving points, and it is often assumed that they can see the positions of other agents.
	For asynchronous rendezvous in graphs, the optimization criterion is the cost, i.e.,  the total number of edge traversals.
	In \cite{BCGIL}, the authors designed almost optimal algorithms for asynchronous rendezvous in infinite multidimensional grids, assuming that an agent knows its position in the grid. In \cite{BBDDP,DPV} this assumption was replaced by a weaker assumption that agents have distinct identities. 

\section{Preliminary Results}\label{prelim}

In this section we prove two preliminary negative results concerning the existence of RV-universal DFA. In order to prove our first result we need the auxiliary notion of an {\em infinite oriented line}. This is an infinite oriented graph whose all nodes have degree 2, and ports at the extremities of each edge, corresponding to this edge, are 0 and 1. At each node, we say that the direction corresponding to port 1 is {\em right} and the direction corresponding to port 0 is {\em left}.  
Consider any DFA  starting at a node of an infinite oriented line in state $S_0$. Let $t_2>t_1$ be the first two rounds when the automaton is in the same state $S$. Let $\tau=t_2-t_1$. Let $v_i$, for $i=1,2$, be the node at which the automaton is in round $t_i$. There are three cases: $v_2=v_1$, $v_2$ is right of $v_1$, and $v_2$ is left of $v_1$. Since the trajectory of the automaton in each period
$(t_1+\tau j,t_1+1+\tau j,\dots,t_2-1+\tau j)$, for any $j>0$, is a shift of its trajectory in the period $(t_1,t_1+1,\dots,t_2-1)$, it follows that in the first case the set of all visited nodes is bounded, in the second case it is a half-line infinite in the direction right, and in the third case it is
a half-line infinite in the direction left. We will say that the DFA is {\em bounded} in the first case, {\em right-oriented} in the second case, and {\em left-oriented} in the third case. Without marking nodes, a $k$-state right-oriented DFA can visit at most $k$ nodes left of its base, and a $k$-state left-oriented DFA can visit at most $k$ nodes right of its base.

We also define a {\em finite oriented line} of length $x$. This is a graph with $x+1$ nodes, two of which have degree 1,  $x-1$ nodes have degree 2, and ports at the extremities of each edge joining nodes of degree 2, corresponding to this edge, are 0 and 1. In a finite oriented line, there is one edge with ports 0, 0 at its extremities, corresponding to this edge. One of these extremities has degree 1. It is called the {\em beginning} of the finite oriented line. The other  node of degree 1 is called the {\em end} of the line.

Finally,  a {\em line} is a graph with two nodes of degree 1 and the remaining nodes of degree 2.

We are now ready to prove our first negative result.

\begin{proposition}\label{prop1}
If the agents cannot mark the nodes in any way, then there does not exist a RV-universal DFA, even for the class of instances whose underlying graph is a line.
\end{proposition}
\begin{proof}
Consider any $k$-state DFA $\cal A$. We will construct a feasible instance whose underlying graph is a line, for which $\cal A$ does not permit to achieve rendezvous.
Consider the behavior of an agent modeled by  $\cal A$, on the infinite oriented line. First suppose that $\cal A$ is bounded and that all nodes visited by $\cal A$ on the infinite oriented line are at distance at most $b$ from the base of $\cal A$. Let $L$ be the finite oriented line of length $4b+3$, let $u$ be the node at distance $b+1$ from its beginning, and let $v$ be the node at distance $b+1$ from its end. Consider the instance $(L,u,v)$ and agents with bases $u$ and $v$, modeled by  $\cal A$.
The behavior of both these agents in $L$ is identical to their behavior on the infinite oriented line. Thus the agents will be always at distance at least 1, hence they cannot meet. On the other hand, the instance $(L,u,v)$ is feasible because there exists a DFA that visits all nodes of the line $L$ and identifies its end. Thus agents can go to the end of the line and meet there.

Thus we can assume that  $\cal A$ is left-oriented or right-oriented. Without loss of generality, assume that the second case holds. Consider the behavior of $\cal A$ on the infinite oriented line, starting at node $z$ in state $S_0$, and let $t_2>t_1$ be the first two rounds when the automaton is in the same state $S$. Let $\tau=t_2-t_1$. Let $v_i$, for $i=1,2$, be the node at which the automaton is in round $t_i$. Let $\Delta>0$ be the distance between $v_1$ and $v_2$. $z$ is at some distance $\delta \leq k$ from $v_1$.

Let $M$ be the finite oriented line of length $2\Delta +2k+1$. Consider two disjoint copies $M'$ and $M''$ of $M$. Let $w_1$ be the node at distance $2k+1$ from the beginning of $M'$ and let  $u$ be the node at distance $\delta$ from $w_1$, on the same side left/right as $z$ with respect to $v_1$.
Let $w_2$ be the node at distance $2k+1+\Delta$ from the beginning of $M''$ and let  $v$ be the node at distance $\delta$ from $w_2$, on the same side left/right as $z$ with respect to $v_1$. Let $N$ be the line resulting  from joining the beginnings of $M'$ and $M''$ by an edge $e$ with both ports 1 at the extremities, corresponding to $e$. The agent $A(u)$ starts at its base $u$, and the agent $A(v)$ starts at its base $v$. Both agents are modeled by $\cal A$. Agent $A(u)$ starts $\tau$ rounds earlier than agent $A(v)$.

The behavior of both agents in instance $(N,u,v)$ is the same as in the infinite oriented line, before each agent visits a
leaf or traverses $e$. Each agent visits a leaf before traversing $e$ and they cannot meet before visiting a leaf. By the construction of the instance, each agent visits a leaf for the first time in the same round and in the same state. Call this round  $t$.

It follows by induction on round $r \geq t$ that in any such round $r$, both agents are in the same state.
In round $t$, the distance between the agents is odd, as the length $2(2\Delta +2k+1)+1$ of the line $N$ is odd.
In any round $r > t$, the distance between the agents may either remain unchanged, or increase by 2 or decrease by 2, with respect to the previous round. Hence this distance remains odd, in any round $r \geq t$, and thus agents can never meet.

It remains to observe that the instance $(N,u,v)$ is feasible. Indeed, an automaton with sufficiently many states can record the distance of its base from the closer leaf. These distances are $2k+1$ and $2k+1+\Delta$, respectively. The agent corresponding to the smaller distance can go to its closer leaf, and the agent corresponding to the larger distance can go to its farther leaf, where the agents will meet. 
\end{proof}

In view of Proposition \ref{prop1}, in order to get RV-universal automata, we need to allow some marking of nodes. In this paper, marking is done by identical pebbles used by the agents. We consider two types of pebbles: {\em stationary pebbles} that can be dropped by agents at nodes but cannot be picked up, and {\em movable pebbles} that can be dropped by agents and later picked up. Our second negative result says that, even in the more powerful scenario of movable pebbles, there is no RV-universal automaton for the class of all instances.

\begin{proposition}\label{prop2}
There does not exist  a RV-universal DFA for the class of all instances, even if each agent is equipped with an arbitrary finite number of movable pebbles.
\end{proposition}
\begin{proof}
Consider any DFA $\cal A$ using a finite number $k$ of movable pebbles. We will construct a feasible instance for which $\cal A$ does not permit to achieve rendezvous.
It is proved in \cite{Ro} that for any finite number of DFA there exists a finite cubic graph $G$ that cannot be collectively explored by these automata, i.e., at least one node of $G$ will not be visited by any of these automata. Since movable pebbles can be simulated by automata, it follows that there exists a finite cubic graph $G$ that cannot be explored by $\cal A$ using finitely many movable pebbles.

Consider the $n$-node graph $G$ and its node $u$ which is the base of an agent modeled by $\cal A$. Let $v$ be the node of $G$ that cannot be visited by $\cal A$ using $k$ pebbles. Let $(G_1,u_1)$ and $(G_2,u_2)$ be disjoint isomorphic copies of $G$ with nodes $u_1$, $u_2$ corresponding to the base $u$. Let $v_1$ and $v_2$ be the nodes in each copy, corresponding to the unvisited node $v$ of $G$.
We construct the graph $H$ as follows. Consider the path $P$ of length $2$ with extremities $v_1$ and $v_2$. Add $n$ leaves adjacent to the mid-point of $P$, creating the augmented path $Q$. The graph $H$ is defined as the union of graphs $G_1$, $G_2$ and $Q$.

The DFA  $\cal A$ does not permit us to achieve rendezvous in the instance $(H,u_1,u_2)$ because agents start at nodes $u_1$ and $u_2$, respectively, and they cannot reach the path $P$ joining the two disjoint copies of $G$ from which they start. It remains to show that the instance $(H,u_1,u_2)$  is feasible. Indeed, an automaton with sufficient memory can explore and remember the entire graph $H$. Thus it can identify the only node $z$ of degree $n+2$ in the graph $H$ (which is the mid-point of $P$ together with the added leaves). After this identification, each agent goes to $z$ and stops. 
\end{proof}

Proposition \ref{prop2} implies that in order to get a RV-universal automaton, we must not only allow pebbles, but also restrict the class of instances for which a RV-universal automaton is sought. Such a natural class are trees, and hence, from now on, we study the existence of
RV-universal automata for the class of instances whose underlying graph is a tree.

\section{Nonexistence of RV-Universal Automata with Stationary\\ Pebbles in Trees}

In this section, we prove that there does not exist any RV-Universal DFA even for the class of all  instances whose underlying graph is a finite line, if the agents modeled by the automaton are initially equipped with finitely many stationary pebbles. 

The high-level idea of the proof is as follows. Given any DFA $A$ with $k$ states, we construct a feasible instance whose underlying graph is a line, and for which $A$ does not guarantee rendezvous.
We first consider two copies of a finite oriented line with beginnings $b$ and $b'$, respectively.
Then we exchange port numbers 0 and 1 at some nodes of degree 2 in each copy, and we call these nodes {\em red}.
All other nodes of degree 2 will be called {\em green}, and the end of each copy will be called {\em blue}.  Finally we join the blue nodes by an edge, putting port 1 corresponding to this edge at each blue node. The agents start at nodes $b$ and $b'$ in the same round.
Red nodes are chosen so that:
\begin{itemize}
\item
they are situated in different places in each copy, to guarantee that the instance is feasible;
\item
all red nodes in each copy are farther from the beginning of the copy  than all pebbles possibly dropped by $A$ in this copy. As we will see, this is done to ensure that (stationary) pebbles cannot help to meet.
The red node in each copy closest to $b$ (resp. to $b'$) is called $c$ (resp. $c'$).
\item
the red nodes farthest from the beginning of each copy, called respectively $g$ and $g'$, are at the same distance from $b$ and $b'$, respectively, and have the property that the agents always enter these nodes in the same round and in the same state.
\end{itemize}
We will prove that the only place where agents could possibly meet is the part of the line between nodes $g$ and $g'$. However, the distance between these nodes is odd, and we will show that in each round after reaching $g$ and $g'$ for the first time, the distance between the agents either remains the same as in the previous round, or increases by 2 or decreases by 2. Hence, when the agents are between $g$ and $g'$, this distance is always odd,
and hence rendezvous is impossible. The rest of this section is devoted to showing that $A$ does not guarantee rendezvous in the line constructed as above.

The aim of our first lemma is to restrict the pool of automata that could be potentially RV-universal.
The lemma examines the behavior of an automaton on any infinite line, i.e., any infinite graph all of whose nodes are of degree 2.

\begin{lemma}\label{finitely}
If a DFA visits a node $v$ of an infinite line more than $k$ times then all nodes that it visits are at distance at most $\mu$ from $v$, for some positive integer $\mu$.
\end{lemma}

\begin{proof}
Consider a DFA $A$ that visits a node $v$ of an infinite line more than $k$ times.  Let $t_2>t_1$ be the first two rounds when the automaton is in the same state $S$ and visits node $v$. Let $\tau=t_2-t_1$. The trajectory of the automaton in each period
$(t_1+\tau i,t_1+1+\tau i,\dots,t_2-1+\tau i)$, for any $i>0$, is identical. Suppose all nodes of this trajectory are at distance at most $\mu'$ from node $v$.
Since  the number of nodes visited  before round $t_1$ is at most $\mu''$, for some positive integer $\mu''$, it follows that all nodes visited by $A$ are at distance at most $\mu=\mu'+\mu''$ from $v$.
\end{proof}

The above lemma permits us to discard all automata that visit some node of an infinite line more than $k$ times. Indeed, agents modeled by such an automaton and starting sufficiently far, e.g., on a finite oriented line, can never meet, as the sets of nodes visited by them are disjoint. Hence, from now on we may assume that, if our given automaton $A$ operates on an infinite line then it visits each node at most $k$ times.

For any node $v$ of an infinite line there are two half-lines determined by $v$.
The next lemma shows that we may restrict attention to automata that, when operating on any infinite line, and given any node $v$ of this line, visit only finitely many nodes of one of the two half-lines determined by $v$.

\begin{lemma}\label{lemma bound}
Consider a DFA $A$ that visits a node $v$ of an infinite line at most $k$ times.
Then there exists a half-line $H$ determined by $v$, and an integer $\beta$, such that $A$ visits at most $\beta$ nodes of $H$. 
\end{lemma}

\begin{proof}
Consider any node $v$ of an infinite line and suppose that the automaton is in the half-line $H'$ after the last visit of $v$. Let $H$ be the other half-line determined by $v$. Let $\beta_0$ be the number of nodes of $H$ visited by $A$ before the first visit  of $v$, and let $\beta_i$, for $1\leq i \leq k$, be the number of nodes of $H$ visited by $A$ between the $(i-1)$th and the $i$th visit of $v$. Then the number of nodes of $H$ ever visited by $A$ is at most $\beta=\max(\beta_0,\dots,\beta_k)$.
\end{proof}

For any node $v$, we will say that the half-line $H' \neq H$ is the {\em free} half-line determined by $v$.

We will need the definition of an infinite {\em pseudo-oriented line}.
This is either an infinite oriented line, or an infinite oriented line, where at one node, called a {\em special node}, the port numbers 0 and 1 are reversed. If a pseudo-oriented line is simply an oriented line, each of its nodes is considered special.

For  a DFA $A$ located at a special node $v$ of a pseudo-oriented line, the free half-line determined by $v$ may be either the half-line for which, at each node $u\neq v$, the port 0 leads towards $v$, or the half-line for which, at each node $u\neq v$, the port 1 leads towards $v$. In the first case we will call $A$ {\em right-sided}, and in the second case we will call it {\em left-sided}. Without loss of generality we may assume from now on that $A$ is right-sided.

Consider the DFA $A$ located at a special node $u$ of an infinite pseudo-oriented line in some round $t_0$.
Suppose that $A$ is in state $S_i$. Let $t_2>t_1\geq t_0$ be the first two rounds when the automaton is in the same state $S'$. Let $a$ denote the distance between the  locations of the automaton on the line in rounds $t_0$ and $t_1$. Define $l(S_i)$ as the distance between the locations of the automaton on the line in rounds $t_1$ and $t_2$ (cf. Fig. \ref{g1}). Also, define $x=2\cdot l(S_1)\cdot l(S_2)\cdots  l(S_k)$, where $S_1,\dots,S_k$ are the states of $A$. 
\begin{figure}[h]

\centering

\includegraphics[width=0.4\columnwidth]{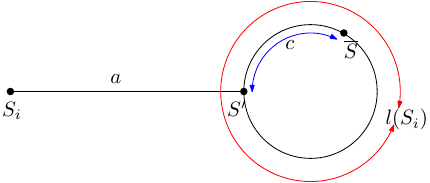}

\caption{\footnotesize State transitions and corresponding trajectories starting from state $S_i$.}
\label{g1}
\end{figure}

Now we are ready to construct a feasible instance for which the $k$-state right-sided automaton $A$ does not guarantee rendezvous. This instance is $(L, b,b')$, for some finite line $L$, whose nodes of degree 1 are $b$ and $b'$. The agents start simultaneously from nodes $b$ and $b'$, respectively. In what follows, we describe the line $L$ (cf. Fig. \ref{g2}).
 
 First, we construct two copies $C$ and $C'$ of a finite oriented line with beginnings $b$ and $b'$, respectively. Since our automaton has only finitely many stationary pebbles at its disposal, we may assume that it does not drop any pebble after some round $t$. Let $e$ and $e'$, respectively, denote the farthest nodes from $b$ and $b'$ in lines $C$ and $C'$, respectively, visited by round $t$.
 
 Let $z$ denote the distance between $b$ (resp. $b'$) and $e$ (resp. $e'$). Next, we 
 color red some nodes of degree 2 of lines $C$ and $C'$. In each red node we exchange the port numbers 0 and 1.
 The closest red nodes $c$ and $c'$ to the beginnings  $b$ and $b'$, in lines $C$ and $C'$, are at distance $z+\beta x$ from $b$ and $b'$, respectively, where $\beta$ is the bound from Lemma \ref{lemma bound}.
The farthest red nodes $g$ and $g'$ from the beginnings $b$ and $b'$, in lines $C$ and $C'$, are at distance $z+6\beta x$ from $b$ and $b'$, respectively. We color red three more nodes between nodes $c$ and $g$ (resp. between nodes $c'$ and $g'$). Let $c_i$ (resp. $c'_i$), for $i=1,2,3$, denote the $i^{th}$ of these red nodes from node $c$ (resp. from node $c'$). 
\begin{itemize}
\item The distance between nodes $b$ and $c_1$ (resp. nodes $b'$ and $c'_1$) is $z+2\beta x$. 
\item The distance between nodes $b$ and $c_2$ is $z+3\beta x$. 
\item The distance between nodes $b'$ and $c_2'$ is $z+4\beta x$. 
\item The distance between nodes $b$ and $c_3$ (resp. nodes $b'$ and $c'_3$) is $z+5\beta x$. 
 \end{itemize}
We say that the line segment between nodes $c$ and $g$ (resp. between nodes $c'$ and $g'$) is the {\it crucial} region of $C$ (resp. $C'$).  Finally, we color blue nodes $f$ and $f'$ at distance $z+6\beta x +1$ from nodes $b$ and $b'$, respectively. All other nodes of degree 2 of lines $C$ and $C'$ are colored green.
This completes the construction of finite oriented lines $C$ and $C'$. Finally, to complete the construction of line $L$, we join blue nodes $f$ and $f'$ by an edge, and put port 1 corresponding to this edge at each blue node. 

\begin{figure}[h]

\centering

\includegraphics[width=0.9\columnwidth]{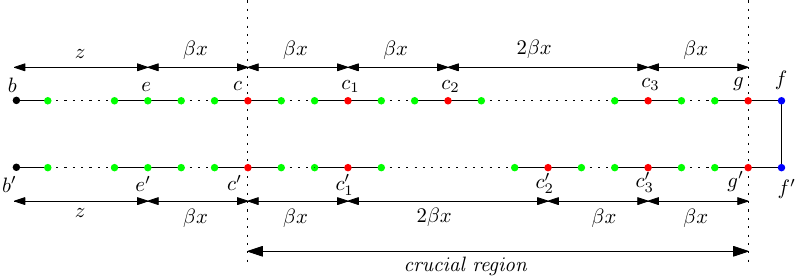}

\caption{\footnotesize Representation of Line $L$.}
\label{g2}
\end{figure}

The next two lemmas consider the behavior of the DFA $A$ on an infinite pseudo-oriented line.  

\begin{lemma}\label{state_repeat}
Let $u$ be a special node of an infinite pseudo-oriented line. Let $A$ be located at node $u$ in state $S$. If $A$ is in state $\overline{S}$ at a node $v$ of the free half-line determined by $u$, at distance $a+c+i\cdot l(S)$ from $u$, for some  $i\geq 0$ and $c\geq 0$, then for all $j\geq i+1$, $A$ will be in the same state $\overline{S}$ at the node $w$ of the free half-line determined by $u$ at distance $a+c+j\cdot l(S)$ from $u$. 
\end{lemma}
\begin{proof}
The proof is by induction on $j$ (cf. Fig. \ref{g1}).\\
\textit{Base Case:} Let $A$ be in state $\overline{S}$ in round $t'_1$ at the node $v$ of the free half-line determined by $u$, at distance $a+c+i\cdot l(S)$ from node $u$, for some $i\geq 0$ and $c\geq 0$. Suppose $A$ is in some state $S^*\neq \overline{S}$ in round $t'_2$ at the node $w$ of the free half-line determined by $u$, at distance $a+c+(i+1)\cdot l(S)$ from $u$. Let $A$ be in some state $S'$ in a round $t_1$, at the node $v_1$ of the free half-line determined by $u$, at distance $a+i\cdot l(S)$ from $u$. Also, let $A$ be in state $S'$ in a round $t_2$, such that $t'_2>t_2>t_1$, at the node $v_2$ of the free half-line determined by $u$, at distance $a+(i+1)l(S)$ from $u$. Consider the sequence of state transitions starting from state $S'$ in the time intervals $[t_1,t'_1]$ and $[t_2,t'_2]$. Let $t$ be the smallest round such that $A$ transitions to different states $S_1$ and $S_2$ from some state $S_3$ in rounds $t_1+t\leq t'_1$ and $t_2+t\leq t'_2$, respectively. This is impossible for a DFA. Thus, the statement of the lemma holds for $j=i+1$.\\
\textit{Inductive Hypothesis:} Suppose the statement of the lemma is true for all $j$, such that $i+1< j\leq m-1$, for some positive integer $m$.\\
\textit{Inductive Step:} Suppose $A$ is in some state $S^*\neq \overline{S}$ in some round $t_3$ at the node $z$ of the free half-line determined by $u$, at distance $a+c+m\cdot l(S)$ from $u$. By the inductive hypothesis, $A$ is in state $\overline{S}$ in rounds $t_5$ and $t_4$, respectively, at nodes $z_1$ and $z_2$ of the free half-line determined by $u$, that are at distances $a+c+(m-2)\cdot l(S)$ and $a+c+(m-1)\cdot l(S)$, respectively, from $u$. Now, compare the sequence of state transitions starting from state $\overline{S}$ in the time intervals $[t_4,t_3]$ and $[t_5,t_4]$. Let $t$ be the smallest round such that $A$ transitions to different states $S_4$ and $S_5$ from some state $S_4$ in rounds $t_4+t\leq t_3 $ and $t_5+t\leq t_4$, respectively. This is impossible for a DFA. Thus, the statement holds for $j=m$.
\end{proof}

\begin{lemma}\label{same_state}
Let $u$ be a special node of an infinite pseudo-oriented line. Let $A$ be located at node $u$ in state $S$. If $A$ is in state $\overline{S}$ at a node $v$ of the free half-line determined by $u$, at distance $\beta x$ from $u$, then $A$ will be in the same state $\overline{S}$ at the node $w$ of the free half-line determined by $u$ at distance $i\beta x$ from $u$, for all $i>1$. 
\end{lemma}
\begin{proof}
We have $\beta x=a+k_1\cdot l(S)+l_1$ and $i\beta x=a+k_2\cdot l(S)+l_2$ where $k_1,k_2\geq 0$, and $0\leq l_1,l_2 <l(S)$. This implies that $(i-1)\beta x=(k_2-k_1)\cdot l(S)+(l_2-l_1)$. Since $\beta x$ is a multiple of $l(S)$, the integer $(k_2-k_1)\cdot l(S)+(l_2-l_1)$ must be a multiple of $l(S)$. Since $0\leq l_1,l_2 <l(S)$, we must have $l_1=l_2$. From Lemma~\ref{state_repeat}, it follows that $A$ will be in the same state $\overline{S}$ at nodes $u$ and $w$ of the free half-line determined by $u$ at distances $\beta x$ and $i\beta x$ from $u$, respectively. 
\end{proof}

We assume that agents $A_1$ and $A_2$ start simultaneously in the state $S_0$ at nodes $b$ and $b'$, respectively. The following lemma shows that the agents will visit the beginnings of the crucial regions at nodes $c$ and $c'$ in the same state and in the same round.

\begin{lemma}\label{same state round}
If the agents start simultaneously at nodes $b$ and $b'$, then they reach nodes $c$ and $c'$ in the same state and in the same round. Moreover, their inputs in this round are equal.
\end{lemma}
\begin{proof}
The proof is by induction on round $t$. We prove that in round $t\geq 0$, before reaching $c$ and $c'$ respectively, the inputs of the agents are the same, and the agents are in the same state. Note that all nodes of degree 2 in the line segment between nodes $b$ and $c$ (resp. $b'$ and $c'$) are green nodes. \\
\textit{Base Case:} Initially, the agents are located at nodes $b$ and $b'$. Both nodes $b$ and $b'$ have degree 1 and are empty. Since the agents start with $\rho$ pebbles in the initial state $S_0$, we have the same input $(1,-1,0,\rho)$ and the same state $S_0$ for both agents in round $t=0$.\\
\textit{Inductive Hypothesis:} Suppose that for all rounds $t\leq r-1$, the inputs of the agents are the same, and the agents are in the same state.\\
\textit{Inductive Step:} By the inductive hypothesis, the agents have the same input $I$, and they are in the same state (say $S'$) in round $r-1$. First, consider the state to which agents transition in round $r-1$. Since the agents are in state $S'$ with the same input $I$, both agents will transition to the state $\delta(S', I)$. 

Next, we consider the agents' outputs in round $r-1$. Since $\lambda (S')$ is identical for the agents, they will have the same value of $m$ in round $r$. Based on the positions of the agents in round $r-1$, we have the following cases:
\begin{enumerate}
\item agents are at nodes $b$ and $b'$: If $j=-1$, the agents remain at nodes $b$ and $b'$ in round $r-1$. Thus, $d=1$, and $i=-1$ for both agents in round $r$. If $j=0$, the agents visit a degree 2 node entering through port 0 from nodes $b$ and $b'$, respectively. Thus, $d=2$ and $i=0$ for both agents in round $r$. Thus the part of the input of both agents concerning the port number and the node degree is the same in round $r$.
\item agents are at some degree 2 nodes. Suppose the agents are at nodes $v$ and $v'$ in round $r-1$. If $j=-1$, the agents remain at nodes $v$ and $v'$ in round $r-1$. Thus, $d=2$, and $i=-1$ for both agents in round $r$. If $j=0$,  and the agents visit a degree 2 node entering through port 1 from nodes $v$ and $v'$, respectively, in round $r$, then $d=2$ and $i=1$ for both agents in round $r$. 
If $j=0$, and the agents visit nodes $b$ and $b'$ by taking port 0 from nodes $v$ and $v'$, respectively, they will enter $b$ and $b'$ through port 0. Thus, $d=1$ and $i=0$ for both agents in round $r$.
If $j=1$, then the agents visit a degree 2 node entering through port 0 from nodes $v$ and $v'$, respectively, in round $r$. Thus, $d=2$ and $i=0$ for both agents in round $r$. 
Thus the agents always have the same part of the input  concerning the port number and the node degree
 in round $r$.
\end{enumerate}
It remains to show that the part of the input  concerning the presence of a pebble is the same for both agents in round $r$.
If $j=0$ or $j=1$, then the agents visit a node other than their location in round $r-1$. If the nodes visited in round $r$ were visited in some round $r'<r-1$, then by the inductive hypothesis, the agents will have the same bit $h$ in round $r$. If the nodes visited in round $r$ were not visited in any previous rounds, then $h=0$ for both agents. Thus the input of both agents  is the same in round $r$.

Therefore, in any round $t\geq 0$, before reaching $c$ and $c'$ respectively, the inputs of the agents are the same, and the agents are in the same state. Hence, if the agents start simultaneously at nodes $b$ and $b'$, then they will reach nodes $c$ and $c'$ in the same state and in the same round. In this round they will have the same input.
\end{proof}

The following lemma shows that if the agents enter the crucial region through nodes $c$ and $c'$ in the same state and in the same round, then they exit the crucial region through nodes $g$ and $g'$, respectively, in the same state and in the same round. 

\begin{lemma}\label{exit crucial}
If the agents enter crucial regions through nodes $c$ and $c'$ in the same state and  in the same  round, then they reach nodes $g$ and $g'$, respectively, in the same state and in the same round. Moreover, the agents have the same input in this round.
\end{lemma}
\begin{proof}
Suppose that the agents $A_1$ and $A_2$ enter their crucial regions through nodes $c$ and $c'$, respectively, in some state $S$ in round $t_0$, for the last time before reaching nodes $g$ and $g'$.
By the definition of nodes $e$ and $e'$ on line $L$, the number $m$  of carried pebbles in the input of agents remains identical for all rounds $t\geq t_0$. Also, $d=2$ and $h=0$ in the inputs at all nodes in the crucial region. Thus, the part of the input of both agents concerning the node degree, the bit $h$, and the number of carried pebbles is the same in any round $t\geq t_0$, in the crucial region. Consider a red node in the crucial region and consider the first visit of this node by an agent. 
We examine the states of the agents, the entry ports, and the rounds at this visit. 
\begin{enumerate}
\item Nodes $c_1$ and $c'_1$: The distance between $c$ and $c_1$ (resp. between $c'$ and $c'_1$) is $\beta x$. The agents visit nodes $c$ and $c'$ in the same round $t_0$, have the same input, and are in the same state $S$, by Lemma \ref{same state round}. Therefore, 
the agents will reach nodes $c_1$ and $c'_1$ for the first time in the same state (say state $S'$), and in the same round (say round $t_1>t_0$). Let $T_1=t_1-t_0$. Since the agents visit nodes $c_1$ and $c'_1$ for the first time, the entry port must be 1. Thus, $A_1$ and $A_2$ have the same input at nodes $c_1$ and $c'_1$, respectively.

\item Nodes $c_2$ and $c'_2$: The distance between $c_1$ and $c_2$ is $\beta x$, and the distance between $c'_1$ and $c'_2$ is $2\beta x$. From Lemma~\ref{same_state}, it follows that $A_1$ and $A_2$ reach nodes $c_2$ and $c'_2$, respectively, in the same state $S'$. 
Let $t_2$ and $t'_2>t_2$ be the rounds in which $A_1$ and $A_2$ reach nodes $c_2$ and $c'_2$, respectively, for the first time. 
Notice that, all nodes visited by agent $A_1$  during the time interval $[t_1+1, t_2-1]$, and all nodes visited by agent $A_2$
during the time interval $[t'_1+1, t'_2-1]$ are green. Hence we can apply Lemma \ref{same_state}, although this lemma assumes that the agent is in a pseudo-oriented line.
Let $T_2=t_2-t_1$ and $T'_2=t'_2-t_1$. Also, let $\delta=T'_2-T_2$. Since the agents visit nodes $c_2$ and $c'_2$ for the first time, the entry port must be 1. Thus, $A_1$ and $A_2$ have the same input at nodes $c_2$ and $c'_2$, respectively.
\item Nodes $c_3$ and $c'_3$: The distance between $c_2$ and $c_3$ is $2\beta x$, and the distance between $c'_2$ and $c'_3$ is $\beta x$. Lemma~\ref{same_state} guarantees that the agents reach nodes $c_3$ and $c'_3$ in the same state $S'$. Let $t_3$ and $t'_3$ be the rounds in which $A_1$ and $A_2$ reach nodes $c_3$ and $c'_3$, respectively, for the first time. Let $T_3=t_3-t_2$ and $T'_3=t'_3-t'_2$. We have $T_3-T'_3=\delta$. Since the agents visit nodes $c_3$ and $c'_3$ for the first time, the entry port must be 1. Thus, $A_1$ and $A_2$ have the same input at nodes $c_3$ and $c'_3$, respectively.
\item Nodes $g$ and $g'$: The distance between $c_3$ and $g$ is $\beta x$, and the distnace between $c'_3$ and $g'$ is $\beta x$.
From Lemma~\ref{same_state}, it follows that the agents reach nodes $g$ and $g'$ in the same state $S'$. Let $t_4>t_3$ and $t'_4>t_3'$ be the rounds in which $A_1$ and $A_2$ reach nodes $g$ and $g'$, respectively, for the first time. Let $T_4=t_4-t_3$ and $T'_4=t'_4-t'_3$. We have $T_4=T'_4$. Since the agents visit nodes $g$ and $g'$ for the first time, the entry port must be 1. Thus, $A_1$ and $A_2$ have the same input at nodes $g$ and $g'$, respectively.
\end{enumerate}
Let $T$ (resp. $T'$) denote the time required to reach node $g$ by $A_1$ (resp. node $g'$ by $A_2$) for the first time after round $t_0$. We have $T=T_1+T_2+T_3+T_4=T_1+T_2+(T'_3+\delta)+T_4=T_1+(T_2+\delta)+T'_3+T_4=T_1+T'_2+T'_3+T_4=T_1+T'_2+T'_3+T'_4=T'$. Thus, $A_1$ and $A_2$ reach nodes $g$ and $g'$ for the first time in the same round $t_4=t'_4$. In this round, they have the same input and are in the same state~$S'$. 
\end{proof}

The next lemma shows that if agents $A_1$ and $A_2$ are at nodes $g$ and $g'$, respectively, in the same state and in the same round, then in all successive rounds the distances of agents $A_1$ and $A_2$ from nodes $f$ and $f'$, respectively, are equal. Let $r$ denote the round in which agents visit nodes $g$ and $g'$ for the first time. For all $t\geq r$, let $D_t$ (resp. $D'_t$) represent the distance between node $f$ (resp. node $f'$) and the location of $A_1$ (resp. $A_2$) in round $t$.
\begin{lemma}\label{lemma distance}
For all rounds $t\geq r$, $D_t=D'_t$. 
\end{lemma}

\begin{proof}
We prove the following invariant by induction on round $t\geq r$.

\begin{itemize}
\item in round $t$, the inputs of the agents are the same and the agents are in the same state.
\item $D_t=D'_t$
\end{itemize}
\textit{Base Case:} In round $t=r$, the agents are located at nodes $g$ and $g'$. From Lemma~\ref{exit crucial}, it follows that the agents have the same input, and they are in the same state in round $r$. We have $D_r=1=D'_r$. \\
\textit{Inductive Hypothesis:} Suppose that the invariant holds for all rounds $t$, such that $t\geq r$ and $t\leq s-1$. \\
\textit{Inductive Step:} By the inductive hypothesis, the agents have the same input $I$, and they are in the same state (say $S'$) in round $s-1$. Moreover, $D_{s-1}=D'_{s-1}$. Since the agents are in state $S'$ with the same input $I$, both agents will transition to the state $\delta(S', I)$. Thus, the agents have the same state in round $s$. 

From Lemma~\ref{lemma bound}, it follows that the agents never visit nodes $c_3$ and $c'_3$ after round $r$. By the definition of nodes $e$ and $e'$ on line $L$, the number $m$  of carried pebbles in the input of agents remains identical for all rounds $t\geq r$. Also, $d=2$ and $h=0$ in the inputs at all nodes for all rounds $t\geq r$. Thus, the part of the input of both agents concerning the node degree, the bit $h$, and the number of carried pebbles is the same in round $s$. We examine the entry port $i$, and the distances $D_s$ and $D'_s$ in round $s$. If $j=-1$, the agents remain in round $s$ at the nodes visited by them in round $s-1$. Thus, $D_{s}=D_{s-1}=D'_{s-1}=D'_s$, and $i=-1$ for both agents in round $s$. 
If $j\neq -1$, we have the following cases, based on the location of agents $A_1$ and $A_2$ in round $s-1$:
\begin{enumerate}

\item $A_1$ is at node $g$ and $A_2$ is at node $g'$ in round $s-1$. We have $D_{s-1}=D'_{s-1}=1$ in round $s-1$. If $j=0$, agents $A_1$ and $A_2$ visit nodes $f$ and $f'$, entering through port 0 from nodes $g$ and $g'$, respectively. Thus, $D_s=D'_s=0$ and $i=0$ for both agents in round $s$. If $j=1$, then agents $A_1$ and $A_2$ visit a green node entering through port 1 from nodes $g$ and $g'$, respectively, in round $s$. Then $D_s=D'_s=2$ and $i=1$, for both agents in round $s$. 

\item $A_1$ is at node $g'$ and $A_2$ is at node $g$ in round $s-1$. We have $D_{s-1}=D'_{s-1}=2$ in round $s-1$. If $j=0$, agents $A_1$ and $A_2$, visit nodes $f'$ and $f$ in round $s$, entering through port 0 from nodes $g'$ and $g$, respectively. Thus, $D_s=D'_s=1$ and $i=0$ for both agents in round $s$. If $j=1$, then agents $A_1$ and $A_2$ visit a green node entering through port 1 from nodes $g$ and $g'$, respectively, in round $s$. Then $D_s=D'_s=3$ and $i=1$, for both agents in round $s$. 

\item $A_1$ is at node $f$ and $A_2$ is at node $f'$ in round $s-1$. We have $D_{s-1}=D'_{s-1}=0$ in round $s-1$. If $j=0$, agents $A_1$ and $A_2$ visit nodes $g$ and $g'$, entering through port 0 from nodes $f$ and $f'$, respectively. Thus, $D_s=D'_s=1$ and $i=0$, for both agents in round $s$. If $j=1$, then agents $A_1$ and $A_2$ visit nodes $f'$ and $f$ entering through port 1 from nodes $f$ and $f'$, respectively, in round $s$. Then $D_s=D'_s=1$ and $i=1$, for both agents in round $s$.

\item $A_1$ is at node $f'$ and $A_2$ is at node $f$ in round $s-1$. We have $D_{s-1}=D'_{s-1}=1$ in round $s-1$. If $j=0$, agents $A_1$ and $A_2$ visit nodes $g'$ and $g$, entering through port 0 from nodes $f'$ and $f$, respectively. Thus, $D_s=D'_s=2$ and $i=0$, for both agents in round $s$. If $j=1$, then agents $A_1$ and $A_2$ visit nodes $f$ and $f'$ entering through port 1 from nodes $f'$ and $f$, respectively, in round $s$. Then $D_s=D'_s=0$ and $i=1$, for both agents in round $s$.

\item $A_1$ and $A_2$ are at some green nodes $v_1$ and $v_2$, respectively. By our construction of line $L$, the distance of a green node from nodes $f$ and $f'$ is at least 2. The following three cases are possible (cf. Fig. \ref{g3}).

\begin{figure}[h]

\centering

{ \includegraphics[width=0.95\columnwidth]{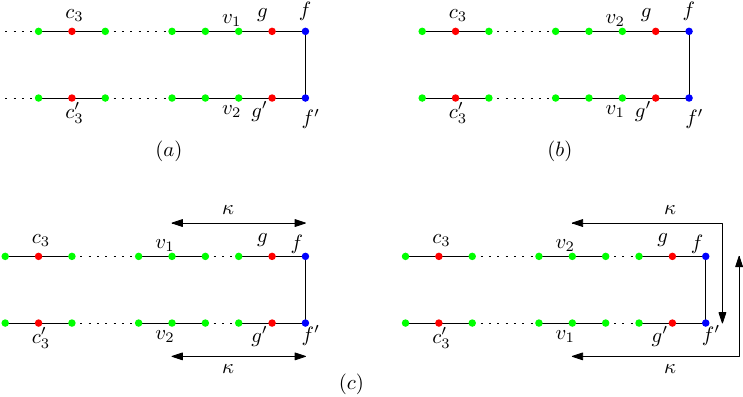}}

\caption{\footnotesize Illustration of subcases (a), (b) and (c) of case 5 in the proof of Lemma \ref{lemma distance}.}
\label{g3}
\end{figure}

\begin{enumerate}
\item nodes $v_1$ and $v_2$ are at distance 2 from nodes $f$ and $f'$, respectively. In this case, nodes $v_1$ and $v_2$ are adjacent to nodes $g$ and $g'$, respectively. We have $D_{s-1}=D'_{s-1}=2$ in round $s-1$. If $j=0$, the agents visit a green node entering through port 1 from nodes $v_1$ and $v_2$, respectively, in round $s$. Thus, $D_s=D'_s=3$ and $i=1$, for both agents in round $s$. If $j=1$, then agents $A_1$ and $A_2$ visit nodes $g$ and $g'$ entering through port 1 from nodes $v_1$ and $v_2$, respectively, in round $s$. Then $D_s=D'_s=1$ and $i=1$, for both agents in round $s$.
\item nodes $v_1$ and $v_2$ are at distance 3 from nodes $f$ and $f'$, respectively. Moreover, nodes $v_1$ and $v_2$ are adjacent to nodes $g'$ and $g$, respectively. We have $D_{s-1}=D'_{s-1}=3$ in round $s-1$. If $j=0$, the agents visit a green node entering through port 1 from nodes $v_1$ and $v_2$, respectively, in round $s$. Thus, $D_s=D'_s=4$ and $i=1$, for both agents in round $s$. If $j=1$, then agents $A_1$ and $A_2$ visit nodes $g'$ and $g$ entering through port 1 from nodes $v_1$ and $v_2$, respectively, in round $s$. Then $D_s=D'_s=2$ and $i=1$, for both agents in round $s$.
\item nodes $v_1$ and $v_2$ are at distance $\kappa\geq 3$ from nodes $f$ and $f'$, respectively. Moreover, nodes $v_1$ and $v_2$ are not adjacent to any of the nodes $g$ and $g'$. We have $D_{s-1}=D'_{s-1}=\kappa$ in round $s-1$. If $j=0$, the agents visit a green node entering through port 1 from nodes $v_1$ and $v_2$, respectively, in round $s$. Thus, $D_s=D'_s=\kappa+1$ and $i=1$, for both agents in round $s$. If $j=1$, the agents visit a green node entering through port 0 from nodes $v_1$ and $v_2$, respectively, in round $s$. Thus, $D_s=D'_s=\kappa-1$ and $i=0$, for both agents in round $s$.
\end{enumerate}
\end{enumerate}
Thus, the agents have the same input, and they are in the same state in round $s$. Moreover, $D_s=D'_s$. 
Hence, for all rounds $t\geq r$, the invariant is proved by induction. 
This proves the lemma.
\end{proof}

We are now ready to show that if agents start simultaneously from the ends of line $L$ then they cannot meet.

\begin{theorem}\label{no meeting}
If the agents start simultaneously at nodes $b$ and $b'$ of line $L$, they can never meet.
\end{theorem}
\begin{proof}
From Lemma~\ref{exit crucial}, it follows that the agents must visit nodes $g$ and $g'$ for the first time in the same state and in the same round $r$. The agents cannot possibly meet before round $r$ because they explore disjoint parts of line $L$ until round $r$. Suppose that the agents meet in some round $r'>r$ at some node $v^*$. From Lemma~\ref{lemma distance}, it follows that $D_{r'}=D'_{r'}$. Thus, the distance between node $v^*$ and node $f$ is equal to the distance between node $v^*$ and $f'$. This is a contradiction because nodes $f$ and $f'$ are adjacent. This proves the theorem.
\end{proof}

\begin{corollary}
There does not exist any RV-Universal DFA even for the class of all instances whose underlying graph is a finite line, if the automaton modelling  the agents is initially equipped with finitely many stationary pebbles. 
\end{corollary}
\begin{proof}
Consider agents modeled by a DFA  $A$ equipped with finitely many stationary pebbles.
Consider the instance $(L, b,b')$, where $L$ is the line constructed above for the automaton $A$ (cf. Fig. \ref{g2}). The agents  modeled by $A$ start simultaneously from nodes $b$ and $b'$. From Theorem~\ref{no meeting}, it follows that the agents can never meet. 

It remains to show that the instance $(L,b,b')$ is feasible. If the agents are modeled by an automaton with a sufficient number of states, then they can explore and remember the entire line $L$. Thus, the agents can identify the difference in the distances from nodes $ b$ and $ b'$ to nodes $ c_2$ and $ c'_2$, respectively, and hence they can distinguish between nodes $b$ and $b'$. Then both agents go to $b$  and stop. 
\end{proof}

\section{RV-Universal Automata with Movable Pebbles in Trees}

In this section we prove that there exists an RV-universal DFA for the class of all instances whose underlying graph is a tree, in the scenario where each of the agents is initially equipped with a single movable pebble. In other words, we prove the existence of a single automaton that guarantees rendezvous for any feasible instance whose underlying graph is a tree, when identical copies of this automaton start with an arbitrary delay, each equipped with a single movable pebble.

As usual in the literature on automata navigating in graphs, 
we will describe the behavior of the agents by designing an algorithm that can be executed by a 
deterministic finite automaton, rather than formally describing the construction of an automaton by defining its output and state transition functions.

Our algorithm is based on the notion of a {\em basic walk} starting at a given node $w$ of the tree. This is the following infinite procedure {\tt BW}$(w)$. The agent starts at node $w$ and takes port 0. In each round, if the agent entered the current node $z$ of degree $d$ by port $j$, then it exits it taking port $(j+1) \mod d$.
Starting at any node $w$ of an $n$-node tree, and executing $2(n-1)$ steps of {\tt BW}$(w)$, the agent visits all nodes of the tree, traverses each edge exactly twice, and gets back to node $w$.

We will use the procedure {\tt ShiftPebble} that will be called during the execution of a basic walk, when the agent is at a pebble. Thus the agent is at some node $y$ and knows which port number $p$ it should take to get  to the next node in the basic walk.\\

{\bf Procedure}  {\tt ShiftPebble}
\nopagebreak

$pick$; take port $p$; $drop$;\\

Throughout this section, $G$ denotes any tree.
We will use the following notion. 
An instance $(G,u,v)$ is called {\em automorphic}, if there is a port preserving automorphism of $G$ carrying $u$ on $v$. More formally, $(G,u,v)$ is automorphic, if there exists a bijection $f: G \longrightarrow G$, such that $f(u)=v$, nodes $x$ and $y$ are adjacent if and only if nodes $f(x)$ and $f(y)$ are adjacent, and if $p$ (resp. $q$) is the port number at $x$ (resp. $y$) corresponding to edge$\{x,y\}$ then $p$ (resp. $q$) is the port number corresponding to edge$\{f(x),f(y)\}$ at $f(x)$ (resp. at $f(y)$).

The main idea of our algorithm is based on the following lemmas.

\begin{lemma}\label{lem1}
 If the instance $(G,u,v)$ is feasible then it is not automorphic.
\end{lemma}

\begin{proof}
Consider a feasible instance $(G,u,v)$ and suppose that it is automorphic.
Let $\cal A$ be a rendezvous algorithm for $(G,u,v)$ and
let $f: G \longrightarrow G$ be a port preserving automorphism of $G$ carrying $u$ on $v$. Consider agents $A(u)$ and $A(v)$ simultaneously starting at bases $u$ and $v$ respectively and executing algorithm $\cal A$. Then in some round they must enter the same node $z$. Let $r$ be the first round when this happens. Since agents execute the same algorithm, by induction on round number $r'<r$, if agent $A(u)$ is at node $x$ in round $r'$ then agent $A(v)$ is at node $f(x)$ in round $r'$. In round $r-1$ they are at distinct nodes $y$ and $f(y)$.
In round $r$ they both take the same port $a$ and hence must enter node $z$ by the same port $b$. 
This is a contradiction, because it is impossible to enter a node from distinct nodes by the same port.
\end{proof}

\begin{lemma}\label{lem2}
Let $(G,u,v)$ be a non-automorphic instance. Consider basic walks {\tt BW}$(u)$ and {\tt BW}$(v)$.
Then there exists some step $i$, such that either the nodes visited in step $i$ in these walks have different degrees, or they have the same degree but the port numbers by which the agents enter these nodes in step $i$ are different.
\end{lemma}

\begin{proof}
The proof is by contraposition. Assume that the conclusion of the lemma does not hold,  i.e., in each step $i$, the following condition $C_i$ is satisfied: the nodes $u_i$ and $v_i$ visited in step $i$ by the basic walks {\tt BW}$(u)$ and {\tt BW}$(v)$, respectively, have equal degree, and port numbers by which the agents enter these nodes in step $i$ are equal. We define $u_0=u$ and $v_0=v$.
We will prove that the instance $(G,u,v)$ is automorphic, by showing that the function $f$ defined by $f(u_i)=v_i$, for all $i\geq 0$, is a port-preserving automorphism of $G$ carrying $u$ on $v$.

We prove the following invariant by induction on $i>0$.
\begin{itemize}
\item
the port number by which the agent leaves node $u_{i-1}$ in the $i$th step of the basic walk {\tt BW}$(u)$ is equal to
the port number by which the agent leaves node $v_{i-1}$ in the $i$th step of the basic walk {\tt BW}$(v)$;
\item
the set of port numbers at node $u_j$, for any $ j\leq i$, used by the agent until step $i$ of {\tt BW}$(u)$ is equal to
the set of port numbers at node $v_j$, for any $ j\leq i$, used by the agent until step $i$ of {\tt BW}$(v)$.
\end{itemize}
The invariant holds for $i=1$ by the fact that a basic walk starts by taking port 0, and by condition $C_1$.
Suppose that the invariant holds for some $i>0$. By condition $C_i$, the agent entered $v_i$ in step $i$ of {\tt BW}$(v)$ by the same port $q$
as it entered $u_i$ in step $i$ of {\tt BW}$(u)$. Since the degrees of $u_i$ and $v_i$ are the same, call them $d$, then the agent
leaves $u_i$ in step $i$ of {\tt BW}$(u)$ by port $(q+1) \mod d$, and  leaves $v_i$ in step $i$ of {\tt BW}$(v)$ also by port $(q+1) \mod d$.
These port numbers must be equal.

Since the exit port numbers from nodes $u_i$ and $v_i$ in step $i$ of {\tt BW}$(u)$ and {\tt BW}$(v)$ are the same, and the entry port numbers to nodes $u_{i+1}$ and $v_{i+1}$ in step $i+1$ of {\tt BW}$(u)$ and {\tt BW}$(v)$ are the same, by condition $C_{i+1}$, the second part of the invariant is satisfied  for $i+1$. This proves the invariant by induction.

Now we are ready to prove that $f$ is a port-preserving automorphism. Suppose that nodes $a$ and $b$ are adjacent and that $p$ is the port number at $a$ and $q$ is the port number at $b$ corresponding to edge $\{a,b\}$. Then edge $\{a,b\}$ is traversed from $a$ to $b$ in some step $i$ of {\tt BW}$(u)$. By the invariant and by condition $C_i$, the agent leaves $f(a)$ by port $p$ and enters $f(b)$ by port $q$ in step $i$ of {\tt BW}$(v)$. Thus nodes $f(a)$ and $f(b)$ are adjacent and corresponding ports are preserved.

It remains to show that $f$ is one-to-one. Suppose that  $a$ and $b$ are two distinct nodes of the tree. Suppose that $b$ is visited for the first time later than $a$ in {\tt BW}$(u)$ and suppose this happens in step $i$ of  {\tt BW}$(u)$. Then, in step $i$ of  {\tt BW}$(u)$, the agent takes a port $p$ at some node $b'$ to reach $b$. By the invariant and by condition $C_i$, in step $i$ of  {\tt BW}$(v)$, the agent takes port $p$ at  node $f(b')$ to reach $f(b)$. Since port $p$ is used for the first time at $b'$ in  {\tt BW}$(u)$, and the set of ports at $b'$ used in  {\tt BW}$(u)$ until step $i-1$ is equal to the set of ports at $f(b')$ used in  {\tt BW}$(v)$ until step $i-1$, it follows that port $p$ is used for the first time at $f(b')$ in  {\tt BW}$(v)$. Hence $f(a)\neq f(b)$.

The fact that $f(u)=v$ follows from the definition. Thus $f$ is a port-preserving automorphism carrying $u$ on $v$, and hence
the instance $(G,u,v)$ is automorphic.
\end{proof}

\subsection{High-level idea of the algorithm}

The high-level idea of our algorithm is as follows. Consider any feasible instance $(G,u,v)$.
By Lemmas \ref{lem1} and \ref{lem2} we know that there exists some step $i$, such that either the nodes visited in step $i$ in the walks {\tt BW}$(u)$ and {\tt BW}$(v)$ have different degrees, or they have the same degree but the port numbers by which the agents enter these nodes in step $i$ are different. Call the above condition an {\em accident}.
Hence, if each agent executes the procedure {\tt BW}$(u)$ (resp. {\tt BW}$(v)$) then an accident must happen in some step. 

Denote by $A(u)$ the agent with base $u$ and by $A(v)$ the agent with base $v$.
At the beginning, each agent drops its pebble at its base.
The algorithm proceeds in phases following the steps of the procedure {\tt BW}. In the $j$th phase, each agent shifts its pebble to the node visited by this agent in step $j$ of its basic walk. The crucial part of the phase consists in visiting both pebbles by both agents, trying to verify if there is an accident in the current phase. Determining if node degrees are different is straightforward, but special care must be taken to determine if the entry port numbers are different. However, the main difficulty is to ensure that at the visit of the pebble of the other agent in phase $j$, this pebble it is still at the node visited by the other agent in step $j$ of its basic walk. This is not straightforward, due to possible initial delay between agents and due to the unknown time of the trip to the other pebble.
Once an accident is discovered by the agents, they can break symmetry considering the lexicographic order
on the couples $(d,q)$, where $d$ is the degree of the node involved in the accident, and $q$ is the entry port number to this node. After breaking symmetry, one agent stays at its pebble and the other completes a basic walk starting from its pebble and meets the idle agent.

\subsection{Detailed description of the algorithm}

We present our algorithm for the automaton $A(u)$. The algorithm runs in phases. In the preliminary phase, $A(u)$ drops its pebble at its base $u$, and executes {\tt BW}$(u)$ without moving the pebble. Note that an automaton cannot remember where its base is located. 
If $A(u)$ starts to move by following its initial basic walk, it will periodically encounter nodes with a pebble. 
The agents can start with an arbitrary delay $\delta$. If $\delta\geq 2(n-1)$, where $n$ is the number of nodes in the tree, then agents will meet by round $2(n-1)$ of the earlier agent because the other agent will be still dormant at its base when the earlier agent gets to this base. Hence, we may assume that $\delta<2(n-1)$.
If $deg(u)\neq deg(v)$, then $A(u)$ identifies the difference at the first visit of a node with a pebble that has different degree than $u$.  In such a case, if the degree of this node is larger than $deg(u)$, then $A(u)$ waits at this node. If the degree of this node is smaller than $deg(u)$, then $A(u)$
continues its initial basic walk until it sees a node of degree $deg(u)$ with a pebble.
This way, the agents will meet in the preliminary phase, at the base of larger degree.
 Otherwise (when $deg(u)=deg(v)$), $A(u)$ terminates its preliminary phase when it encounters a node with a pebble $2deg(u)+1$ times during the execution of {\tt BW}$(u)$.

If there is no rendezvous in the preliminary phase, then for all $j\geq 1$, $A(u)$ starts its phase $j$ at its pebble. In particular, it starts phase 1 at its base $u$.
$A(u)$ performs the following actions in phase $j$:
\begin{enumerate}
\item Suppose $u_{j}$ denotes the current position of the pebble. Note that in phase 1, $u_{j}=u$. $A(u)$ executes {\tt BW}$(u)$. $A(u)$ records the outgoing port number from node $u_j$ in this walk (say $p$), and the entering port number to the next node (say node $u_j'$) reached through port $p$ from node $u_j$. Call this entering port number $q$. Consider the next node with a pebble in the basic walk {\tt BW}$(u)$ from which the agent exits by port $p$.
Call this node $y$.
Let $z$ be the node entered from $y$
 by taking port $p$ during the execution of {\tt BW}$(u)$. Suppose $r$ denotes the entering port number to node $z$ in {\tt BW}$(u)$. If $q<r$, then $A(u)$ waits at node $z$ forever. If 
$q>r$, then $A(u)$ continues {\tt BW}$(u)$ until reaching a node with a pebble for the $(2\cdot deg(u_j)+1)^{th}$ time, and waits there forever. If $q=r$, then $A(u)$ continues the execution of {\tt BW}$(u)$ until visiting a node with a pebble for the $2\cdot (2deg(u_j)+1)^{th}$ time. During the execution of {\tt BW}$(u)$, if $A(u)$ visits a non-leaf node $\nu$ between two consecutive visits of a node with a pebble, then $A(u)$ goes to node $\nu$ and waits there forever.

\item $A(u)$ shifts its pebble to the next node in {\tt BW}$(u)$. Note that  $u_j'$ denotes the new position of the pebble of $A(u)$. $A(u)$ continues executing {\tt BW}$(u)$. Suppose $t_1$ denotes the round in which $A(u)$ shifts its pebble. Let $w$ be the first visited node with a pebble such that $deg(w)\neq deg(u_j')$ after round $t_1$. If $deg(u_j')<deg(w)$, then $A(u)$ waits at the node $w$ forever.
If $deg(u_j')>deg(w)$, then $A(u)$ continues {\tt BW}$(u)$ until reaching a node of degree $deg(u_j')$ with a pebble, and waits there forever. 
If a node with a pebble of a degree different than $deg(u’_j)$ is not
encountered, then the agent continues execution of basic walk {\tt BW}$(u)$ until
visiting a node with a pebble for the
$2\cdot(2deg(u’_j)+ 1)^{th}$ time after round $t_1$.
Then it terminates phase $j$ and starts phase $j+1$.
\end{enumerate}
Note that $A(u)$ starts its phase $j\geq 1$ at the node visited in step $j-1$ of its basic walk {\tt BW}$(u)$ and ends this phase at the node visited in step $j$ of its basic walk {\tt BW}$(u)$. The algorithm is interrupted as soon as the agents meet.
The pseudocode of the algorithm, executed by an agent  $A(u)$ is presented below.

\begin{algorithm}
\footnotesize

\tcp{** Preliminary Phase **//}
Drop a pebble at $u$ and record $deg(u)$\;
Execute {\tt BW}$(u)$ until visiting a node $v$ with a pebble of degree different than $deg(u)$ or until visiting a node $v$ with a pebble  $2deg(u)+1$ times\;
\uIf{$deg(u)<deg(v)$}
{Wait at $v$ forever\;
}
\If{$deg(u)>deg(v)$}
{Continue {\tt BW}$(u)$ until reaching a node of degree $deg(u)$ with a pebble, and wait there forever;}
\tcp{** End of Preliminary Phase **//}
{\For {$j=1,2,\ldots$}
{
\tcp{** Phase $j$ **//}

$u_{j}:=$ the current position of the pebble; \tcp{** In Phase 1, $u_{j}=u$ **//}
Execute {\tt BW}$(u)$\;
 
$p:=$ the outgoing port from node $u_j$ during the execution of {\tt BW}$(u)$\; 
$u'_j:=$ the node entered by taking port $p$ from node $u_j$\;
 $q:=$ the entering port to node $u'_j$\;
  
 $y:=$ the next node with a pebble in the basic walk {\tt BW}$(u)$ from which the agent exits by port $p$\;
$r:=$ the entering port number to node $y$\;
 \uIf{$q<r$}
 {Wait at node $y$ forever\;} 
 \If{$q>r$}{Continue {\tt BW}$(u)$ until visiting a node with a pebble for $(2\cdot deg(u_j)+1)$ times, and wait there forever\;}
 Continue the execution of {\tt BW}$(u)$ until visiting a node with a pebble for $2\cdot (2deg(u_j)+1)$ times\;
 \uIf {there is a single non-leaf node $\nu$ between two consecutive visits of a node with a pebble during the execution of {\tt BW}$(u)$ }
 {go to node $\nu$ and wait forever}
 \Else
 {Execute {\tt ShiftPebble}\;}
$u'_j:=$  the new position of the pebble\;
  $t_1:=$ the round in which $A(u)$ shifts its pebble\;
 Continue the execution of  {\tt BW}$(u)$ until visiting a node $w$ with a pebble, of degree different than $deg(u_j')$ or until visiting a node $w$ with a pebble $2\cdot(2deg(u’_j)+ 1)$ times\;
 \uIf{$deg(u'_j)<deg(w)$}
{Wait at the node $w$ forever\;
}
\If{$deg(u_j')>deg(w)$}
{Continue {\tt BW}$(u)$ until visiting a node of degree $deg(u'_j)$ with a pebble, and wait there forever;}

}
}  
     
\caption{{\tt RendezvousAlgorithm} for agent $A(u)$}
\label{algo1}
\end{algorithm}

\subsection{Proof of correctness}

An agent starting at any node $w$ of an $n$-node tree, and executing $2(n-1)$ steps of {\tt BW}$(w)$, visits all nodes of the tree, traverses each edge exactly twice, and gets back to node $w$. However, the agent cannot possibly identify a complete traversal of the tree for unknown $n$. Since an agent is equipped with a pebble, if it drops its pebble, and executes $2(n-1)$ steps of {\tt BW}$(w)$ under the assumption that there are no other pebbles, then it could identify a complete traversal of the tree by counting the number of times it visits a node with a pebble, which is $2deg(w)$. As another agent is also navigating in the network, and will also drop a pebble at some node $v\neq w$, a complete traversal of the tree is identified by $deg(w)+deg(v)+1$ visits of a node with a pebble. If a node with a pebble of a degree different than $deg(w)$ is not encountered, i.e., $deg(w)=deg(v)$, then an agent can identify a complete traversal of the tree by counting $2 deg(w)+1$ visits of a node with a pebble. 

The following lemma ensures that at the end of each phase $j\geq 1$ by the earlier agent, each agent's pebble is located at the node visited in step $j$ of its  basic walk. 

\begin{lemma}\label{samenode}
For each $j\geq 1$, at the end of phase $j$ of the earlier agent, each agent's pebble is located at the node visited in step $j$ of its basic walk.
\end{lemma} 
\begin{proof} 
The proof is by induction on $j$. Without loss of generality, assume that $A(u)$ is the earlier agent.\\
\textit{Base Case:} Let $t_1$ and $t'_1=t_1+\delta$ denote the rounds in which agents $A(u)$ and $A(v)$ start phase $j=1$, respectively. At the beginning of phase 1, the pebble of $A(u)$ (resp. the pebble of $A(v)$) is located at base $u$ (resp. base $v$). Let $t_2$ (resp. $t'_2$) denote the round in which $A(u)$ (resp. $A(v)$) shifts its pebble to the node visited in step $j$ of its basic walk. Let $t_3$ (resp. $t'_3$) denote the round in which $A(u)$ (resp. $A(v)$) terminates phase 1. We have $t'_2=t_1+\delta+2\cdot 2(n-1)+1<t_1+3\cdot 2(n-1)+1$ and $t_3=t_1+2\cdot 2(n-1)+1+2\cdot 2(n-1)=t_1+4\cdot 2(n-1)+1>t'_2$. Thus, the statement is true for $j=1$. \\
\textit{Inductive Hypothesis:} Assume that the statement is true for $j=k$, i.e., at the end of phase $k$ of the earlier agent, each agent's pebble is located at the node visited in step $k$ of its basic walk. \\
\textit{Inductive Step:} Let $u_k$ (resp. $v_k$) denote the node visited in step $k$ of the basic walk of 
 $A(u)$ (resp. $A(v)$). By inductive hypothesis,
 at the beginning of phase $k+1$, $u_k$ (resp.  $v_k$) is the location of the pebble of $A(u)$ (resp. $A(v)$).
Let $t_4$ and $t'_4=t_4+\delta$ denote the rounds in which agents $A(u)$ and $A(v)$ start phase $k+1$, respectively. Let $t_5$ (resp. $t'_5$) denote the round in which $A(u)$ (resp. $A(v)$) shifts its pebble to the node visited in step $k+1$ of its initial basic walk. Let $t_6$ (resp. $t'_6$) denote the round in which $A(u)$ (resp. $A(v)$) terminates phase $k+1$. We have $t'_5=t_4+\delta+2\cdot 2(n-1)+1<t_4+3\cdot 2(n-1)+1$ and $t_6=t_4+2\cdot 2(n-1)+1+2\cdot 2(n-1)=t_4+4\cdot 2(n-1)+1>t'_5$. Thus, the statement holds for $j=k+1$.
\end{proof}

During the execution of {\tt RendezvousAlgorithm}, an agent shifts its pebble in each phase $j\geq 1$ by following its basic walk. If one of the agents starts the execution of {\tt RendezvousAlgorithm} with a delay of $\delta<2(n-1)$, then it will start each phase $j\geq 1$ with the same delay $\delta$. Let $t$ denote the round in which agent $A(u)$ shifts its pebble in phase $j$. When can it be guaranteed that, after round $t$, $A(u)$ visits a node with the other agent's pebble after the pebble's position has been shifted in phase $j$? The following lemma answers this question.

\begin{lemma}\label{visit}
Let $t$ denote the round in which agent $A(u)$ has shifted its pebble in phase $j\geq 1$. Then agent $A(u)$ visits a node with the other agent's pebble after the position of the pebble has been shifted by the other agent in phase $j$, within $4(n-1)$ steps of the basic walk of $A(u)$ in phase $j$ after round $t$.
\end{lemma}
\begin{proof}
Let $t_1$ and $t'_1=t_1+\delta$ denote the rounds in which the agents start phase $j\geq 1$. First, assume that $A(u)$ is the earlier agent. From Lemma~\ref{samenode}, it follows that at the beginning of phase $j$, $A(u)$ and its pebble (resp. $A(v)$ and its pebble) are located at node $u_j$ (resp. node $v_j$). 
Let $A(u)$ (resp. $A(v)$) shift its pebble from node $u_j$ to node $u'_j$ (resp. from node $v_j$ to node $v'_j$) in round $t_2$ (resp. round $t'_2$). We have $t_2=t_1+2\cdot 2(n-1)+1$, and $t'_2=t_2+\delta$. Suppose $A(u)$ visits the node $v_j$ for the last time within $4(n-1)$ steps of {\tt BW}$(u)$ after round $t_2$ in some round $t$. We have $t>t_2+2(n-1)$ because, starting from round $t_2$, after $2(n-1)$ steps of {\tt BW}$(u)$, $A(u)$ will return to $u'_j\neq v_j$. 
Assume that $t\leq t'_2$. We have $t\leq t'_2=t_2+\delta<t_2+2(n-1)$ which contradicts the inequality $t>t_2+2(n-1)$. Thus, $t>t'_2$. This concludes the proof in the case when $A(u)$ is the earlier agent.
Next, assume that $A(u)$ is the later agent. $A(v)$ shifts its token in round $t_2$, and $A(u)$ shifts its pebble in round $t'_2=t_2+\delta$. Thus, at $A(u)$'s first visit of the other agent's pebble after round $t'_2$ in phase $j$, this pebble has already been shifted in round $t_2$ by the other agent. This first visit occurs within $2(n-1)$ steps of the basic walk of $A(u)$ in phase $j$ after round $t'_2=t$.
\end{proof}

The following lemma concerns the case in which the degrees of the nodes visited in step $j$ of the agents' respective basic walks differ. 
\begin{lemma}\label{diff_degree}
If the degrees of the nodes visited by the agents in step $j$ of their respective basic walks differ, then the agents meet by the end of phase $j\geq 1$.
\end{lemma}
\begin{proof}
Let $j$ be the smallest positive integer such that the degrees of the nodes visited by the agents in step $j$ of their respective basic walks are different. Let $t_1$ and $t'_1=t_1+\delta$ denote the rounds in which the agents start phase $j$. Without loss of generality, assume that $A(u)$ is the earlier agent. Let $A(u)$ (resp. $A(v)$) shift its pebble from node $u_j$ to node $u'_j$ (resp. from node $v_j$ to node $v'_j$) in round $t_2$ (resp. round $t'_2$). Also, the nodes where pebbles are located in rounds $t_2+1$ and $t'_2+1$ are the nodes visited by agents $A(u)$ and $A(v)$ in step $j$ of their respective basic walks. From Lemma~\ref{visit}, it follows that within $4(n-1)$ steps of {\tt BW}$(u)$ after round $t_2+1$, $A(u)$ will visit the node $v'_j$. If $deg(u'_j)<deg(v'_j)$, then $A(u)$ waits at node $v'_j$ forever, and $A(v)$ will return to $v'_j$ where they will meet. If $deg(u'_j)>deg(v'_j)$, $A(u)$ continues {\tt BW}$(u)$ until visiting a node with a pebble  $(deg(u'_j)+deg(v'_j)+1)$ times and waits there forever, while $A(v)$ will reach this node to complete rendezvous.
\end{proof}

The following two lemmas concern the case in which the degrees of the nodes visited in step $j$ of the agents' respective basic walks are equal. 

\begin{lemma}\label{eq_degree}
Assume that the degrees of nodes visited by the agents in step $j$ of their respective basic walks are equal.
Then, by the end of phase $j\geq 1$, the agents compare the entry ports to the nodes they visited in step $j$ of their respective basic walks. 
\end{lemma}
\begin{proof}
Let $t_1$ and $t'_1=t_1+\delta$ denote the rounds in which the agents start phase $j\geq 1$.
In these rounds, the agents record the entry port numbers to nodes visited by each of them in step $j$ of their basic walks.
 Without loss of generality, assume that $A(u)$ is the earlier agent. Let $u_j$ (resp. $v_j$) denote the node where the pebble of $A(u)$ (resp. $A(v)$) is located in round $t_1$ (resp. round $t'_1$). Let $u'_j$ (resp. $v'_j$) denote the node reached from node $u_j$ (resp. node $v_j$) by following the basic walk. Note that $u'_j$ and $v'_j$ are the nodes visited in step $j$ of the agents respective basic walks. We have $deg(u'_j)= deg(v'_j)$.

First, consider the round $t_2$ in which agent $A(u)$ checks the entry port to node $v'_j$.  Let $T_1=t_1+1+t_2$. $A(v)$ will shift its pebble in round $T_2=t'_1+1+4\cdot 2(n-1)$. We have $T_1\leq t_1+1+ 2(n-1)<t'_1+1+4\cdot 2(n-1)=T_2$. Thus, $A(u)$ correctly compares the entry ports to nodes visited by the agents in step $j$ of their respective basic walks. 
 
 Next, consider the round $t_3$ in which agent $A(v)$ checks the entry port to node $u'_j$. Let $T_3=t'_1+1+t_3$. $A(u)$ will shift its pebble in round $T_4=t_1+1+4\cdot 2(n-1)$. We have $T_3=t'_1+1+t_3\leq t'_1+1+ 2(n-1)=t_1+\delta+1+2(n-1)<t_1+1+2\cdot 2(n-1)=T_4$. Thus, $A(v)$ correctly compares the entry ports to nodes visited by the agents in step $j$ of their respective basic walks. 
\end{proof}

\begin{lemma}\label{eq_meet}
If the entry ports to the nodes visited by the agents in step $j$ of their respective basic walks are different, then the agents meet by the end of phase $j\geq 1$.
\end{lemma}
\begin{proof}
Let $k$ be the smallest positive integer such that the entry ports to the nodes visited by the agents in step $k$ of their respective basic walks are different. Let $u'_j$ (resp. $v'_j$) denote the nodes visited by the agents in step $j$ of their respective basic walks. Let $q$ (resp. $r$) denote the entry ports of node $u'_j$ (resp. node $v'_j$) in the agents' basic walks.
From Lemma~\ref{eq_degree}, it follows that in phase $j$, agents compare the entry ports to nodes visited in step $j$ of their respective basic walks. If $q<r$, then $A(u)$ waits at node $v'_j$ forever, and $A(v)$ will return to $v'_j$ where they will meet. If $q>r$, $A(u)$ continues {\tt BW}$(u)$ until visiting a node with a pebble for $(deg(u'_j)+deg(v'_j)+1)$ times and waits there forever, while $A(v)$ will reach this node to complete rendezvous.
\end{proof}

We are now ready to prove the main result of this section.

\begin{theorem}
If the instance $(G,u,v)$ is feasible, then agents meet within $\mathcal{O}(n^2)$ rounds from the start of the execution of {\tt RendezvousAlgorithm} by the earlier agent. 
\end{theorem}
\begin{proof}
During the execution of {\tt RendezvousAlgorithm}, each agent shifts the pebble it dropped at its base in the preliminary phase. The shift is by one step, following the basic walk in each phase $j\geq 1$. In the preliminary phase, each agent executes $2\cdot 2(n-1)$ steps of its basic walk starting from its base. In each phase $j\geq 1$, each agent executes $4\cdot 2(n-1)$ steps of its basic walk, starting from the node visited in step $j-1$ of its basic walk. Since the instance is feasible, it follows from Lemmas \ref{lem1} and \ref{lem2} that
there exists some step $j$ such that either the nodes visited in step $j$ of the basic walks of the agents have different degrees, or they have the same degree but the port numbers by which the agents enter these nodes in step $j$ differ. In the preliminary phase, the agents compare the degrees of bases. 
In each phase $j\geq 1$, the agents compare the entry ports to nodes visited by them in step $j$ of their respective basic walks, and they compare the degrees of the nodes visited in step $j$ of their respective basic walks. During the execution of their respective basic walks in phase $j$, if the agents visit a non-leaf node $\nu$ between two consecutive visits of a node with a pebble then they go to node $\nu$ and stop.

There are two cases. 
If the entry ports to nodes $u_j'$ and $v_j'$ are different then agents meet by the end of phase $j$, in view of Lemma~\ref{eq_meet}. If
$deg(u_j')\neq deg(v_j')$ then agents meet by the end of phase $j$, in view of Lemma \ref{diff_degree}. In both cases, $j\leq 2(n-1)$ because $2(n-1)$ is the number of steps of the basic walk guaranteeing the full traversal of the tree. On the other hand, each phase lasts at most
$8(n-1)$ rounds.
Thus, agents meet within $\mathcal{O}(n^2)$ rounds from the start of the execution of {\tt RendezvousAlgorithm} by the earlier agent.
\end{proof}

Since {\tt RendezvousAlgorithm} requires remembering only a bounded number of bits, we have the following corollary.

\begin{corollary}
There exists an RV-universal DFA for the class of all instances whose underlying graph is a tree, if  each of the agents is initially equipped with a single movable pebble.
\end{corollary}

\section{Conclusion}

We first observed that, even in the more powerful scenario of movable pebbles, if agents are equipped with any finite number of pebbles, there is no RV-universal DFA for the class of all instances. Hence we restricted attention to instances where the underlying graph is a tree, and proved two contrasting results concerning the existence of RV-universal DFA for trees, depending on whether pebbles are stationary or movable.

However, our initial observation concerning the non-existence of RV-universal DFA for the class of all instances, with any number of movable pebbles, uses a counterexample involving quite strange graphs constructed in \cite{Ro} that cannot be collectively explored by any number of DFA. 
Thus it is natural to ask if the existence of a
universal automaton is refuted only by very strange graphs.
More precisely, how large is the class of instances for which there exists a RV-universal DFA using movable pebbles. We proved that this class  contains at least all instances whose underlying graph is a tree. It is an open question if this class is significantly larger.



\begin{thebibliography}{12}

\bibitem{alpern02b}
S. Alpern and S. Gal,
The theory of search games and rendezvous.
Int. Series in Operations research and Management Science,
Kluwer Academic Publisher, 2002.

%
%


\bibitem{BCGIL}
E. Bampas, J. Czyzowicz, L. Gasieniec, D. Ilcinkas, A. Labourel, Almost optimal asynchronous rendezvous in infinite multidimensional grids,
Proc. 24th International Symposium on Distributed Computing (DISC 2010),  297-311.


%

%


\bibitem{BBDDP}
S. Bouchard, M. Bournat, Y. Dieudonn\'{e}, S. Dubois, F. Petit,
Asynchronous approach in the plane: a deterministic polynomial algorithm, 
Distributed Computing 32 (2019), 317-337.

 


\bibitem{CFPS}
M. Cieliebak, P. Flocchini, G. Prencipe, N. Santoro, 
Distributed computing by mobile robots: Gathering, SIAM J. Comput. 41 (2012), 829-879.


\bibitem{CCGKM}
A. Collins, J. Czyzowicz, L. Gasieniec, A. Kosowski, R. A. Martin,
Synchronous rendezvous for location-aware agents. 
Proc. 25th International Symposium on Distributed Computing (DISC 2011), 447-459.


\bibitem{CKP}
J. Czyzowicz, A. Kosowski, A. Pelc, How to meet when you forget: Log-space rendezvous in arbitrary graphs, Distributed Computing 25 (2012), 165-178. 

\bibitem{CKP1}
Time vs. space trade-offs for rendezvous in trees,
Proc. 24th ACM Symposium on Parallelism in Algorithms and Architectures  (SPAA 2012), 1-10.






%

\bibitem{DP}
B. Das, A. Pelc,
Universal Deterministic Symmetry Breaking Between Anonymous Agents in Networks,
Proc. 38th ACM Symposium on Parallelism in Algorithms and Architectures (SPAA 2026), 503-523.


\bibitem{DFKP}
A. Dessmark, P. Fraigniaud, D. Kowalski, A. Pelc.
Deterministic rendezvous in graphs.
Algorithmica 46 (2006), 69-96.

\bibitem{DPV}
Y. Dieudonn\'{e}, A. Pelc, V. Villain, How to meet asynchronously at polynomial cost, 
 SIAM Journal on Computing 44 (2015), 844-867. 
 

\bibitem{fpsw}
P. Flocchini, G. Prencipe, N. Santoro, P. Widmayer,
Gathering of asynchronous robots with limited visibility, Theoretical Computer Science 337 (2005), 147-168.

\bibitem{FP}
P. Fraigniaud, A. Pelc, Delays induce an exponential memory gap for rendezvous in trees, ACM Transactions on Algorithms 9 (2013), 17:1-17:24. 

\bibitem{GP}
Y. Gao, A. Pelc, Gathering teams of deterministic finite automata on a line, Proc. 28th Conference on Principles of Distributed Computing (OPODIS 2024), 11:1 - 11:17. 


\bibitem{KKSS}
			E. Kranakis, D. Krizanc, N. Santoro and C. Sawchuk, 
			Mobile agent rendezvous in a ring, 
			Proc. 23rd Int. Conference on Distributed Computing Systems
			(ICDCS 2003), 592-599.


\bibitem{MP}
A. Miller, A. Pelc, Fast deterministic rendezvous in labeled lines, Proc. 37th International Symposium on Distributed Computing (DISC 2023), 29:1 - 29:22. 


\bibitem{Pe2}
A. Pelc, Deterministic rendezvous algorithms, in: Distributed Computing by Mobile Entities, P. Flocchini, G. Prencipe, N. Santoro, Eds., Springer 2019, LNCS 11340. 

\bibitem{PY}
A. Pelc, R. Yadav, Using Time to Break Symmetry: Universal Deterministic Anonymous Rendezvous,
Proc. 31st ACM Symposium on Parallelism in Algorithms and Architectures (SPAA 2019) 85-92


%

\bibitem{Ro}
H. Rollik, Automaten in planaren Graphen, Acta Informatica 13 (1980) 287–298 (also in Lecture Notes of
Computer Science, Vol. 67, 1979, pp. 266–275).


\bibitem{TSZ07}
A. Ta-Shma and U. Zwick.
Deterministic rendezvous, treasure hunts and strongly universal exploration sequences,
	ACM Trans. Algorithms 10 (2014), 12:1-12:15.







	
	

\end{thebibliography}
\end{document}